\documentclass[conference,10pt]{IEEEtran}
\usepackage{times,color,pst-all,url,subfigure,graphicx}
\usepackage{amsmath,amssymb,amsfonts,amsthm}

\newcommand{\Xnote}[1]{[\textcolor{pink}{\textit{#1}}]}
\newcommand{\SDnote}[1]{[\textcolor{green}{\textit{#1}}]}

\newcommand{\Prob}[1]{{\ensuremath{\mathbb{P}\left[ #1 \right]} }}
\newcommand{\Expc}[1]{{\ensuremath{\mathbb{E}\left[ #1 \right]} }}

\newtheorem{lemma}{Lemma}

\newcommand{\deltaE}{\tilde{\delta}_{E}}
\newcommand{\key}{\mathcal{S}}
\newcommand{\numnodes}{n}
\newcommand{\numintf}{m}
\newcommand{\eff}{E}
\newcommand{\msg}{\mathcal{M}}
\newcommand{\initkey}{\sigma}
\newcommand{\auth}{\mathit{auth}}

\title{Exchanging Secrets Without Using Cryptography 
\vspace{-0.4cm}}
\author{Iris Safaka, Mahdi J Siavoshani, Uday Pulleti, Emre Atsan, \\
Christina Fragouli, Katerina Argyraki, Suhas Diggavi \\
(EPFL and UCLA)
\vspace{-0.3cm}}

\begin{document}
\maketitle

\begin{abstract}
We consider the problem where a group of $\numnodes$ nodes, connected to the same broadcast channel (e.g., a wireless network),
want to generate a common secret bitstream, in the presence of an adversary Eve, who tries to obtain information on the bitstream.
We assume that the nodes initially share a (small) piece of information, but do not have access to any out-of-band channel.
We ask the question: can this problem be solved without relying on Eve's computational limitations,
i.e., without using any form of public-key cryptography?

We propose a secret-agreement protocol, where the $\numnodes$ nodes of the group keep exchanging bits until they have all 
agreed on a bit sequence that Eve cannot reconstruct with very high probability.
In this task, the nodes are assisted by a small number of interferers, whose role is to create channel noise in a way 
that bounds the amount of information Eve can overhear.
Our protocol has polynomial-time complexity and requires no changes to the physical or MAC layer of network devices. 

First, we formally show that, under standard theoretical assumptions, our protocol is information-theoretically secure,
achieves optimal secret-generation rate for $\numnodes = 2$ nodes, and scales well to an arbitrary number of nodes.
Second, we adapt our protocol to a small wireless $14$ m$^2$ testbed;
we experimentally show that, if Eve uses a standard wireless physical layer 
and is not too close to any of the nodes,
$8$ nodes can achieve a secret-generation rate of $38$ Kbps.
To the best of our knowledge, ours is the first experimental demonstration of information-theoretic
secret exchange on a wireless network at a rate beyond a few tens of bits per second.
\end{abstract}

\section{Introduction}

We consider the problem where a group of $\numnodes$ nodes, connected to the same broadcast channel (e.g., a wireless network),
want to generate a common secret bitstream, in the presence of an adversary Eve, who tries to obtain information on the bitstream.
We assume that the $\numnodes$ nodes initially share a (small) piece of information,
but do not have access to any out-of-band channel or public-key infrastructure.
Today, this problem can be solved using public-key cryptography, such as the
Diffie-Hellman~\cite{DiH76} or the RSA \cite{RSA78,last} key-agreement protocols:
the nodes could use the initially shared information to form public keys,
then periodically use a key-agreement protocol to generate new secrets.
These protocols make the assumption that Eve
does not have the computational resources to perform certain operations, such as large-integer factorization.
We ask the question: can we solve this problem {\em without} making any assumption about the adversary's computational capabilities,
i.e., without using public-key (or any classic form of) cryptography?

Our work is based on the following observation: whereas one node's transmission may be overheard by multiple nodes, if
there is noise in the broadcast channel, it becomes unlikely that any two nodes 
overhear exactly the same bits.
To illustrate, consider a person speaking in a low voice in a crowded room---people standing nearby will hear parts of the speech,
but if the room is sufficiently noisy, it is unlikely that any two of these people 
(standing at different locations) will hear the exact same parts of what the speaker says.
We build upon this observation and use channel noise to prevent eavesdroppers from hearing exactly the same bits
as honest participants.  We propose a new secret-agreement protocol, where the $\numnodes$ honest nodes keep exchanging
bits until they have all agreed on a bit sequence that Eve, with very high probability, cannot reconstruct.  
In this task, the $\numnodes$ nodes are assisted by \emph{interferers}
that create the right kind and level of noise to bound the amount of information that Eve overhears.

Unlike existing cryptographic solutions to this problem, which rely on the fact that the adversary cannot compute some function,
our protocol relies on the fact that the adversary cannot overhear {\em all} the information received by an honest node.
In other words, we shift the challenge from computation to network presence:
for cryptographic approaches, a dangerous adversary is one with high computational power (e.g., one with access to quantum computers);
for our approach, a dangerous adversary is one who is physically present in many locations in the network at the same time.

We do not imply that there is any pressing need to replace the existing crypto-systems 
that rely on the adversary's computational limitations.
However, we believe that exploring alternative approaches (which rely on different kinds of adversary limitations)
will become of increasing interest in the near future, as governments and corporations acquire massive
amounts of computational capabilities, and as new technologies, such as quantum computing, mature.  
This interest is already present in the industry community, where several companies are developing
quantum key distribution (QKD) systems~\cite{C1}; these enable a pair of nodes 
to exchange a bitstream that is information-theoretically secure, i.e., an adversary that observes the exchange
obtains $0$ information on the bitstream, independently from her computational capabilities.
A typical commercial application envisioned for QKD systems is the periodic generation of one-time pads at a high enough rate
to enable information-theoretically secure transmission of real-time video, e.g., for military operations~\cite{Qmil}.
On the other hand, QKD systems are expensive due to the need for sophisticated equipment, such as photon detectors.
This motivated us to start exploring the feasibility of secrecy that does not rely on computational
limitations, but also does not require expensive equipment.

In this paper, we make one theoretical and one practical contribution:

(1) On the theory side, we design a protocol that enables
$\numnodes \ge 2$ nodes, which are connected to the same broadcast channel and initially share a small common secret $\initkey$,
to establish a common secret bitstream.
Assuming a typical information-theory model (independent erasure channels between the nodes and known erasure probabilities),
we formally show that:
(i) our protocol is information-theoretically secure; 
(ii) it achieves a secret-generation rate that is optimal for $\numnodes = 2$ nodes
and scales well to an arbitrary number of nodes $\numnodes$.
(iii) it requires per-node operations that are of polynomial complexity.
To the best of our knowledge,
we advance the state of the art in information-theoretic security by (i) allowing groups
to have an arbitrary number of nodes $\numnodes$, and (ii) providing
a polynomial-time protocol that is implementable even in the simplest
wireless devices without any changes to their MAC or physical layers.

(2) On the practical side, 
we adapt our protocol to a small wireless testbed that covers an area of $14$ m$^2$ and consists of
$8$ trusted nodes, $6$ interferers (each with two directional antennae), and an adversary. 
We experimentally show that, in our testbed,
if the adversary uses a standard wireless physical layer\footnote{We argue that our protocol 
will work well even if the adversary possesses a custom physical layer, 
but we have not showed this experimentally.}
 and is located within no less than $1.75$ m from any trusted node,
the $\numnodes=8$ trusted nodes achieve
a common secret generation rate of $38$ secret kilobits per second.
This is in contrast to the state of the art in deploying information-theoretical security on wireless networks, 
which achieves (only pair-wise) secret-generation rates on the order of a few tens of bits per second.
To the best of our knowledge, ours is the first experimental demonstration of information-theoretic
secret exchange on a wireless network at rates of kilobits per second.

We think that this is an important first step toward bridging the gap between theoretical and practical work in this area:
on the one hand, we confirm that the criticism typically addressed at theoretical work on information-theoretic security
(independent erasure channels and predictable erasure probabilities are not realistic network conditions) is valid;
on the other hand, we demonstrate that it is feasible to emulate these conditions in a real wireless network,
enough to generate an information-theoretically secure bitstream of tens of kilobits per second.
That said, we still need to complete several steps before our protocol is ready for non-experimental deployment:
we need to harden it against collusion and wireless denial-of-service attacks;
and we need to find a more practical and energy-efficient way of creating interference than using dedicated, 
trusted interfering nodes (discussion in Section~\ref{sec:discuss}).

In the rest of the paper, we summarize the main idea behind our work (\S\ref{sec:setup}),
we describe our basic secret-agreement protocol, which is information-theoretically secure against a passive adversary (\S\ref{sec:protocol}),
state its properties (\S\ref{sec:analysis}), and describe how to secure the protocol against active adversaries (\S\ref{sec:authenticate}).
Next we describe how we adapted our protocol to a small wireless testbed (\S\ref{sec:real})
and present our experimental results (\S\ref{sec:jam}).
Finally, we discuss limitations and challenges (\S\ref{sec:discuss})
and related work (\S\ref{sec:related}), and we conclude (\S\ref{sec:conclusion}).

\section{Setup}
\label{sec:setup}

\subsection{Problem Statement}

We consider a set of $\numnodes \ge 2$ trusted nodes, $T_0, \ldots, T_{n-1}$, which share an initial common secret $\initkey$
and are connected to the same broadcast channel;
we will refer to these nodes as \emph{terminals};
sometimes we will refer to terminals $T_0$, $T_1$, and $T_2$ respectively as Alice, Bob, and Calvin.
We also assume that we have available $\numintf$ trusted \emph{interferers}, i.e., nodes that can generate noise.

We consider an adversary, Eve, who is connected to the same broadcast channel as the terminals, but has $0$ information on $\initkey$.
When we say that Eve acts as a ``passive'' adversary, we mean that she does not perform any transmissions, i.e.,
she only tries to eavesdrop on the terminals' communications;
when we say that Eve acts as an ``active'' adversary, we mean that she may perform transmissions, e.g., try to impersonate a terminal.
As a first step, in this paper, we assume that Eve only has access to a standard physical layer, i.e., no custom hardware.
Also we assume that there is no collusion, i.e., there may be multiple adversaries like Eve, and each of them may move
around the network, but they may not exchange information.
In Section~\ref{sec:discuss} (where we discuss the limitations of our proposal),
we describe how we believe that these issues can be addressed.

Our goal is to design a protocol that enables the $\numnodes$ terminals to establish a new common secret $\key$
of size $|\key| \gg |\initkey|$, such that Eve (whether acting as a passive or active adversary) obtains $0$ information on $\key$. 
By periodically running this protocol, the terminals will be able to establish a common secret bitstream.

The terminals communicate with each other in two ways:
(1) When we say that terminal $T_i$ \emph{transmits} a packet, we mean that it broadcasts the packet once.
(2) When we say that terminal $T_i$ \emph{reliably broadcasts} a packet,
we mean that it ensures that all other terminals $T_i$ receive it, e.g., through acknowledgments and retransmissions;
to be conservative, we assume that Eve receives all reliably broadcast packets.\footnote{Our use of reliable broadcasting
is similar to the use of a ``public channel'' in information theory, with the difference that we do not assume
that this channel has zero cost, i.e., when we compute protocol efficiency, we do take into account the transmissions
made using this channel.}

We assume a packet-erasure channel between each pair of nodes:
when $T_i$ transmits a packet, $T_j$ (or Eve) receives the entire packet correctly with probability $1-\delta$,
or does not receive it at all; $\delta$ is the \emph{erasure probability} of the channel between $T_i$ and $T_j$ (or Eve).

We summarize the most commonly used symbols in Table~\ref{tab:symbols}.

\subsection{Main Idea}
\label{sec:setup:main}

We define the ``theoretical'' network conditions as follows:
\begin{enumerate}
\item
The channel between Alice and any other terminal $T_i$ is independent from the channel between Alice and Eve,
i.e., when Alice transmits a packet, the event that $T_i$ receives it is independent from the event that Eve does.
\item
We know the erasure probability $\delta_E$ of the Alice-Eve channel.
\end{enumerate}

Suppose, for the moment, that these theoretical conditions hold and that we have only $n=2$ terminals, Alice and Bob.
The erasure probability of the Alice-Bob channel is $\delta_1 = 0.5$
and of the Alice-Eve channel $\delta_E = 0.4$, i.e., Eve has better connectivity to Alice than Bob.

Alice wants to establish a secret $\key$ with Bob. 
To this end, Alice transmits $10$ packets, $x_1,x_2,\ldots, x_{10}$.
Bob correctly receives $5$ of them, $x_1,x_3,x_5,x_7,x_9$. 
Eve correctly receives $6$ of the transmitted packets, $x_1,x_3,x_5,x_6,x_8,x_{10}$.
At this point, there exist $2$ packets, $x_7,x_9$, whose contents are known to both Alice and Bob, but not to Eve.  
Alice does not know which are these packets, but she can guess \emph{how many} they are:
First, she learns which packets Bob received (which is easily accomplished by having Bob send a feedback message to Alice,
specifying the identities of the packets he received).
Combining this piece of information with the erasure probability between Alice and Eve, $\delta_E$,
Alice can compute the expected number of packets received by Bob but not Eve;
in our example, $\delta_E = 0.4$, hence, of the $5$ packets that Bob received, Eve is expected to not have received $0.4 \cdot 5 = 2$ packets.

Alice and Bob establish a secret by performing privacy amplification:
Alice creates $2$ linear combinations of the packets that Bob received, 
$s_1=x_1 \oplus x_5 \oplus x_9$ and $s_2=x_3 \oplus x_7$ (where $\oplus$ denotes the bit-wise XOR operation over the payloads
of the corresponding packets), and creates the secret as their concatenation, $\key = \langle s_1, s_2 \rangle$.
Then, Alice reliably broadcasts to Bob which packets to combine to also create the secret.
Even though Eve receives the reliable broadcast and learns the identities of the packets that were combined to create the secret,
she cannot compute either $s_1$ or $s_2$, because she does not know the contents of $x_7$ and $x_9$. 

\subsection{Challenges}
\label{sec:setup:challenges}

To turn the above idea into a secret-agreement protocol, we answer the following questions:

{\bf (1) Choosing the right linear combinations.}
In the above example, Alice chose $2$ particular linear combinations, whose values Eve could not compute, 
i.e., Alice performed privacy amplification using very simple linear operations. 
But is it always possible for Alice to  choose  linear combinations of the packets received by Bob, such that
Eve cannot compute their values (without knowing which packets Eve received)?

{\bf (2) Extending to multiple nodes.}
Suppose that we have not two, but three terminals, Alice, Bob, and Calvin, that need to establish a common secret $\key$.
Of course, once we have a scheme that works for two terminals, Alice could use it to establish a secret, $\key_B$, with Bob,
a separate secret, $\key_C$, with Calvin, then secretly communicate to each of Bob and Calvin a common secret $\key$.
However, intuitively, we should be able to do better---surely, Bob and Calvin must have received a common subset of Alice's packets.
How can we leverage this common subset to build a better protocol,
and how do we generalize it to an arbitrary number of terminals $\numnodes$?

We address questions (1) and (2) in Section~\ref{sec:protocol}, where we present a secret-agreement
protocol between $\numnodes$ nodes, which assumes theoretical network conditions and passive adversaries.

{\bf (3) Deterring active attacks.}
In the above example, Eve only tried to eavesdrop on Alice's and Bob's communications.
In reality, she could also try to impersonate them, e.g., pretend to be Alice and initiate the secret exchange with Bob.
So, how do we prevent impersonation by adversaries?

We address question (3) in Section~\ref{sec:authenticate}, where we
describe how to add authentication to our secret agreement to protect
it against active adversaries.

{\bf (4) Emulating the theoretical network conditions.} 
In the above example, Alice and Bob were able to establish a secret, because
the Alice-Bob channel was independent from the Alice-Eve one, which ensured that
Bob received some packets that Eve did not receive.
In reality, we cannot assume this:
if Eve positions herself close to Bob, it is quite likely that she will receive all the packets that he receives.
The next question then is, can we artificially condition a real wireless network, so as to ensure
that Bob (or any trusted node) receives some minimum number of packets that Eve does not receive?

{\bf (5) Estimating what the adversary received.} 
In the above example, Alice created a secret by concatenating $2$ linear combinations, 
exactly as many as the number of packets received by Bob but not Eve.
This is the longest secret that Alice and Bob could establish without revealing any information about the secret to Eve:
suppose they attempted to generate a longer secret, by creating and concatenating $3$ linear combinations of the packets received by Bob,
e.g., $\key = \langle s_1, s_2, s_3 \rangle$, where $s_1 = x_1 \oplus x_5 \oplus x_9$, 
$s_2 = x_3 \oplus x_7$, and $s_3 = x_3 \oplus x_5 \oplus x_7 \oplus x_9$;
but then, Eve would learn something about the secret---that $s_1 \oplus s_2 \oplus s_3 = x_1$, 
which would increase her chances of guessing the secret.
The gist is that, since there are only $2$ packets received by Bob but not Eve, there exist exactly $2$ linear combinations
of the packets received by Bob that do not reveal any information to Eve.
Hence, it is important for Alice to correctly guess how many packets were received by Bob but not Eve:
if she overestimates this number, she will create a longer secret than she should, i.e., reveal some information about the secret to Eve;
if she underestimates this number, she will create a shorter secret than she could, i.e., achieve a lower secret-generation rate than
the maximum possible.
So, the last question is, how do we estimate how many packets Eve received in order to create the right secret size?

We address questions (4) and (5) in Section~\ref{sec:jam}, where we adapt our secret-agreement protocol to a 
small wireless testbed.

\subsection{Background}

We now summarize the background needed to read the rest of the paper.
Readers who are familiar with linear coding and information-theoretic secrecy
can skip this section.

Consider two packets, $x_1$ and $x_2$, of the same length.
When we say that we \emph{linearly combine these packets over the field $\mathbb{F}_{2^s}$}, we mean that:
(1) We take the contents of packet $x_1$ and interpret every $s$ consecutive bits as a symbol;
as a result, we get a sequence of symbols, $\langle p_1, p_2, \ldots \rangle$.
(2) We do the same for $x_2$, such that we get another sequence of symbols, $\langle q_1, q_2, \ldots \rangle$.
(3) We perform some linear operation, denoted by $\oplus$, over each pair of symbols $p_i$ and $q_i$;
the resulting sequence, $\langle p_1 \oplus q_1, p_2 \oplus q_2, \ldots \rangle$, is the outcome of the linear combination of the two packets.
For example, if our field is $\mathbb{F}_{2}$, i.e., $s=1$, then we interpret each bit as a symbol;
in that case, performing a linear combination is equivalent to bitwise XOR-ing the contents of the two packets. 

The secrecy of a system $\mathcal{S}$ can be expressed as a function of its entropy, $H(\mathcal{S})$, 
which is a measure of an eavesdropper's uncertainty about the system.
By observing the system, the eavesdropper may learn something about it;
this is captured by the concept of \emph{conditional entropy}, $H(\mathcal{S}|\mathcal{O})$,
which expresses the eavesdropper's uncertainty about the system $\mathcal{S}$, 
after she has completed her observation, $\mathcal{O}$.
A system $\mathcal{S}$ is \emph{information-theoretically secure}, when, by
observing $\mathcal{S}$, an eavesdropper does not decrease her uncertainty about it,
i.e., $H(\mathcal{S}|\mathcal{O}) \rightarrow H(\mathcal{S})$.

In our context, the ``observer'' is Eve, the ``system'' is the secret $\key$,
and Eve's ``observation'' is the set of packets transmitted by the terminals that Eve receives, $\mathcal{X}_E$.
In the worst case, after making her observation, Eve knows the value of $\key$, in which case $H(\key|\mathcal{X}_E) \rightarrow 0$.
In the best case, after making her observation, Eve knows only the length $|\key|$ of the secret, i.e., to her, 
$\key$ is equally likely to have any value from $0$ to $2^{|\key|-1}$;
in this case, $H(\key|\mathcal{X}_E) \rightarrow -\log \frac{1}{|\key|}$, which is the maximum uncertainty that Eve can ever have about $\key$.
So, we say that our secret agreement is information-theoretically secure, 
when, by receiving the set of packets $\mathcal{X}_E$, 
Eve does not decrease her uncertainty about the secret, i.e., 
$H(\key|\mathcal{X}_E) \rightarrow H(\key) = -\log \frac{1}{|\key|}$.

\subsection{Quality Metrics}

We use two metrics to characterize the quality of our secret-agreement protocol:

{\bf Reliability.}
$R = \frac{H(\key|\mathcal{X}_E)}{H(\key)},$
where $\key$ is the secret generated by the protocol, $\mathcal{X}_E$ is the set of packets transmitted by the terminals 
and correctly received by Eve, 
$H(\key)$ is $\key$'s entropy from Eve's point of view, and $H(\key|\mathcal{X}_E)$ is $\key$'s conditional entropy
from Eve's point of view, after Eve has received $\mathcal{X}_E$.
Reliability $R$ means that Eve can correctly guess each secret bit of $\key$ with probability $2^{-R}$,
hence the entire secret $\key$ with probability $2^{-R|\key|}$.
In the example of Section~\ref{sec:setup:main},
Eve learns nothing about $\key$ after receiving packets $x_1, x_3, x_5, x_6, x_8, x_{10}$, hence the reliability of that
particular secret agreement is $1$.

{\bf Efficiency.}
$\eff = \frac{\mathit secret\;\;size}{\mathit transmitted\;\;bits}.$
This captures the cost of our protocol, i.e., the
amount of traffic it produces in order to generate a secret of a given length.
Maurer has formally shown that, assuming the theoretical network conditions, the maximum efficiency that can be achieved when we have two terminals, Alice and Bob,
is $\eff = \delta_E (1 - \delta_1)$, where $\delta_E$ and $\delta_1$ are Eve's and Bob's erasure probabilities, respectively
(although he has not shown \emph{how} to achieve this upper bound using operations of bounded complexity)~\cite{Ma-IT93}.
In the example of Section~\ref{sec:setup:main},
Alice sends $10$ packets in order to establish a secret whose length corresponds to $2$ packets;
hence, ignoring, for the moment, the feedback sent by Bob to Alice, 
the efficiency of that particular secret agreement is $\eff = 0.2$, 
which is the maximum possible according to Maurer's upper bound.

Ideally, we want our protocol to have reliability $1$
and efficiency that is equal to the known optimal for $\numnodes = 2$ terminals
and scales well to an arbitrary number of terminals.

\begin{table}[t]
\centering
\begin{tabular}{|l|l|}
\hline
{\bf Symbol} & {\bf Meaning}  \\
\hline
$\numnodes$ & Number of terminals \\
$\numintf$ & Number of trusted interferers \\
$T_i$ & Terminal $i$ \\
$\delta_E$ & Erasure probability of Alice-Eve channel \\
\hline
\multicolumn{2}{|c|}{Parameters used in both the Basic and Adapted protocols} \\
\hline
$N$ & Number of $x$-packets transmitted by Alice \\
    & (initial phase, step $1$) \\
$N^*$ & Number of $x$-packets received by at least \\
    & one terminal (initial phase, step $1$) \\
$N_i$ & Number of $x$-packets received by $T_i$ \\
    & (initial phase, step $1$) \\
$M$ & Number of $y$-packets created in the initial phase \\
    & (initial phase, step $3$) \\
$M_i$ & Number of $y$-packets reconstructed by $T_i$ \\
    & (initial phase, step 4) \\
$L$ & Number of $s$-packets created by Alice \\
    & (reconciliation phase, step $3$ (Basic) or step $4$ (Adapted)) \\
\hline
\multicolumn{2}{|c|}{Parameters used only in the Adapted protocol} \\
\hline
$M^j$ & Number of $y$-packets created by terminal $T_j$ \\
    & in the Adapted protocol (step 1, initial phase) \\
$K_i$ & Number of $z$-packets created by terminal $T_i$ \\
    & in the Adapted protocol (step 2, reconciliation phase) \\
\hline
\end{tabular}
\caption{\label{tab:symbols} Commonly used symbols}
\vspace{-1cm}
\end{table}

\section{Basic Secret-Agreement Protocol}
\label{sec:protocol}

In this section, we describe a secret-agreement protocol that:
given $\numnodes$ trusted nodes, it allows them to create a common secret $\key$.
We will show that, assuming the theoretical network conditions, a \emph{passive} 
adversary obtains $0$ information about $\key$.

\subsection{Gist}

Alice first transmits $N$ packets.
Suppose that, of these, $\hat{N}$ are commonly received by all the terminals,
hence $\delta_E \hat{N}$ are received by all the terminals but not Eve.
At this point, Alice could create a secret by creating $\delta_E \hat{N}$ combinations 
of the $\hat{N}$ commonly received packets
(as she did in the example of Section~\ref{sec:setup:main});
however, $\hat{N}$ decreases exponentially with the number of terminals, such that $\delta_E \hat{N}$ quickly goes to $0$.
Instead, Alice transmits a second round of packets,
with the purpose of increasing the amount of information that is commonly known to her and all the other terminals,
without increasing, as much as possible, the amount of information known to Eve.

Hence, our protocol consists of two phases.
In the \emph{initial} phase, Alice transmits $N$ packets, which results in her sharing
some number of (different) secret packets with each terminal.
In the \emph{reconciliation} phase, Alice transmits additional information, which results in her sharing
the \emph{same} secret packets with \emph{all} terminals.
I.e., the reconciliation phase does not increase the number of secret packets shared by Alice and any terminal,
just ``redistributes'' them, such that all terminals share the same number of secret packets.

The basic structure of the two phases is similar:
Alice first transmits some number of packets (e.g., $x_1, \ldots x_{10}$);
she creates linear combinations of these packets (e.g., $y_1, y_2$) and tells the other terminals
how she created each combination (e.g., that $y_1 = x_1 \oplus x_5 \oplus x_9$ and $y_2 = x_3 \oplus x_7$);
each terminal reconstructs as many linear combinations as it can (depending on which initial packets it received).
The point of this exchange is always to ``mix'' the information shared by the terminals,
such that, even if Eve has overheard some of the initial packets, she still has no
information on the linear combinations.

\subsection{Basic Protocol Description}

\paragraph*{Initial Phase}
\begin{enumerate}
\item
Alice transmits $N$ packets (we will call them \emph{$x$-packets}).
\item
Each terminal $T_{i\neq 0}$ reliably broadcasts a feedback message specifying which $x$-packets it received. 
\item 
Alice creates $M$ linear combinations of the $x$-packets
(we will call them \emph{$y$-packets}), as described in ``$y$-packet construction'' below.
She reliably broadcasts the coefficients she used to create the $y$-packets.
\item 
Each terminal $T_{i \neq 0}$ reconstructs as many (say $M_i$) of the $y$-packets as it can.
\end{enumerate}

\smallskip
At this point, Alice shares $M_i$ $y$-packets with each terminal $T_i$.
If $\numnodes =2$ terminals, the common secret is the concatenation of 
the $M_1$ $y$-packets shared with $T_1$, $\key = \langle y_1, \ldots, y_{M_1} \rangle$,
and the protocol terminates.

\smallskip
\paragraph*{Reconciliation Phase}
\begin{enumerate}
\item 
Alice creates $M - \min_i M_i$ linear combinations of the $y$-packets (we will call them \emph{$z$-packets}),
as described in ``$z$-packet construction'' below.
She reliably broadcasts both the contents and the coefficients of the $z$-packets.
\item
Each terminal $T_{i\neq 0}$ reconstructs all the $M$ $y$-packets by combining the $M_i$ $y$-packets 
it reconstructed in step $4$ of the initial phase with ${M-M_i}$ of the $z$-packets.
\item 
Alice creates $L = \min_i M_i$ linear combinations of the $y$-packets
(we will call them \emph{$s$-packets}), using the construction specified in Lemma~\ref{lem:linear_3} (Appendix, Section~\ref{app:1}).
She reliably  broadcasts the coefficients she used to create all the $s$-packets. 
\item
Each terminal $T_{i\neq 0}$ reconstructs all the $s$-packets. 
\end{enumerate}

\smallskip
At this point, Alice shares the same $L = \min_i M_i$ $s$-packets with each terminal $T_i$.
The common secret is the concatenation of these $s$-packets,
$\key = \langle s_1, \ldots, s_L \rangle$, and the protocol terminates.

\smallskip
\paragraph*{$y$-packet construction}
Alice identifies the $N^*$ $x$-packets that were received by at least one terminal.
She considers each subset of terminals $\mathcal{J}$, identifies the $N^\mathcal{J}$ $x$-packets that
were received by all the terminals in the subset but no other terminals, and
creates $\delta_E N^\mathcal{J}$ linear combinations of these packets using the construction specified
in Lemma~\ref{lem:linear_1} (Appendix, Section~\ref{app:1}).
As a result, she creates $\delta_E N^*$ linear combinations.
For an illustration, see Figure~\ref{fig_connections}.

\smallskip
\paragraph*{$z$-packet construction}
Alice chooses the $z$-packets such that: every terminal $T_{i\neq0}$ can combine $M-M_i$ $z$-packets with the
$M_i$ $y$-packets it reconstructed in step $4$ of the initial phase, and reconstruct all the $M$ $y$-packets;
choosing the $z$-packets can be done using standard network-coding techniques \cite{mono}.

\begin{figure*}[!t]
\begin{center}
\psset{unit=0.045in}
\begin{pspicture}(-12,-27)(130,115)

\psset{linewidth=1.0pt}

\psline[linewidth=.1]{-}(-18,109)(35,109)
\psline[linewidth=.1]{-}(65,109)(136,109)
\rput(50,109){{\em  Initial Phase, Step 1}}
\rput(50,104){{\bf Alice} transmits $N$ $x$-packets}

\rput(0,99){{\bf Bob} receives}
\rput(0,94){$N_1 = |\mathcal{X}_1|+|\mathcal{X}_{12}|$ $x$-packets}  

\rput(40,99){{\bf Calvin} receives}
\rput(40,94){$N_2 = |\mathcal{X}_2|+|\mathcal{X}_{12}|$ $x$-packets}  

\rput(80,99){{\bf Eve} receives}
\rput(80,94){$N_E$ $x$-packets}  

\rput(8,80){$\mathcal{X}_1$}\rput(19,80){$\mathcal{X}_{12}$}\rput(30,80){$\mathcal{X}_2$}
\rput(2,68){\psellipse[linewidth=0.5mm,linestyle=dashed,hatchangle=0](12,12)(11,7)}
\rput(12,68){\psellipse[linewidth=0.5mm,hatchangle=30](12,12)(11,7)}

\rput(70,84){$\mathcal{X}_1$: $x$-packets received by Bob, not Calvin}
\rput(70,79){$\mathcal{X}_2$: $x$-packets received by Calvin, not Bob}
\rput(70,74){$\mathcal{X}_{12}$: $x$-packets received by Bob and Calvin}

\psline[linewidth=.1]{-}(-18,65)(35,65)
\psline[linewidth=.1]{-}(65,65)(136,65)
\rput(50,65){{\em  Initial Phase, Step 3}}
\rput(50,60){{\bf Alice} constructs $M$ $y$-packets,}
\rput(65,55){ $|\mathcal{Y}_1| = \delta_E  |\mathcal{X}_1|$ from $\mathcal{X}_1$,}
\rput(70,51){ $|\mathcal{Y}_2| = \delta_E  |\mathcal{X}_2|$ from $\mathcal{X}_2$,}
\rput(75,47){ $|\mathcal{Y}_{12}| = \delta_E  |\mathcal{X}_{12}|$ from $\mathcal{X}_{12}$.}

\psline[linecolor=darkgray,linestyle=dashed]{->}(8,72)(2,54)
\psline[linecolor=darkgray,linestyle=dashed]{->}(19,72)(17,54)
\psline[linecolor=darkgray,linestyle=dashed]{->}(28,72)(32,54)

\rput(-5,42){\psellipse[linewidth=0.5mm,linestyle=dashed,hatchangle=0](7,7)(6,4.5)}
\rput(10,42){\psellipse[linewidth=0.5mm,linestyle=dotted,hatchangle=0](7,7)(5,4)}
\rput(26,42){\psellipse[linewidth=0.5mm,linestyle=solid,hatchangle=0](7,7)(6,4.5)}

\rput(2,49){$\mathcal{Y}_1$}\rput(17,49){$\mathcal{Y}_{12}$}\rput(33,49){$\mathcal{Y}_2$}

\psline[linewidth=.1]{-}(-18,40)(35,40)
\psline[linewidth=.1]{-}(65,40)(136,40)
\rput(50,40){{\em  Initial Phase, Step 4}}

\rput(0,35){{\bf Bob} reconstructs}
\rput(0,30){$M_1 = \delta_E  N_1$ $y$-packets}  

\rput(40,35){{\bf Calvin} reconstructs}
\rput(40,30){$M_2 = \delta_E  N_2$ $y$-packets}  

\rput(80,35){{\bf Eve} reconstructs}
\rput(80,30){$M_E$ $y$-packets}

\psline[linewidth=.1]{-}(-18,22)(30,22)
\psline[linewidth=.1]{-}(70,22)(136,22)
\rput(50,22){{\em  Reconciliation Phase, Step 1}}
\rput(50,17){{\bf Alice} reliably broadcasts $M - M_1$ linear combinations of the $y$-packets}

\psline[linewidth=.1]{-}(-18,11)(30,11)
\psline[linewidth=.1]{-}(70,11)(136,11)
\rput(50,11){{\em  Reconciliation Phase, Step 2}}

\rput(0,6){{\bf Bob} reconstructs}
\rput(0,1){all the $M$ $y$-packets}  

\rput(40,6){{\bf Calvin} reconstructs}
\rput(40,1){all the $M$ $y$-packets}  

\rput(80,6){{\bf Eve} reconstructs}
\rput(80,1){$M_E'$ $y$-packets}

\psline[linewidth=.1]{-}(-18,-5)(30,-5)
\psline[linewidth=.1]{-}(70,-5)(136,-5)
\rput(50,-5){{\em  Reconciliation Phase, Step 3}}
\rput(50,-10){{\bf Alice} constructs $M_1$ $s$-packets}

\psline[linewidth=.1]{-}(-18,-16)(30,-16)
\psline[linewidth=.1]{-}(70,-16)(136,-16)
\rput(50,-16){{\em  Reconciliation Phase, Step 4}}

\rput(0,-21){{\bf Bob} reconstructs}
\rput(0,-26){all the $M_1$ $s$-packets}  

\rput(40,-21){{\bf Calvin} reconstructs}
\rput(40,-26){all the $M_1$ $s$-packets}  

\rput(80,-21){{\bf Eve} reconstructs}
\rput(80,-26){$M_E$ $s$-packets}

\psline[linewidth=.1]{-}(100,115)(100,-27)

\rput(120,113){Parameter values}

\rput(120,93){$N_1 = N_2 \rightarrow (1-\delta)  N$ } 
\rput(120,88){$N_E \rightarrow (1-\delta_E)  N$ } 

\rput(120,60){ $M= \delta_E  (1 - \delta^2)  N$ }

\rput(120,35){$M_1= M_2 = \delta_E N_1$}
\rput(120,30){$\rightarrow \delta_E  (1 - \delta)  N$}
\rput(120,25){$M_E \rightarrow 0$}


\rput(120,1){$M_E' \rightarrow M - M_1$}
\rput(120,-10){$M_1 \rightarrow \delta_E  (1 - \delta)  N$}
\rput(120,-25){$M_E \rightarrow 0$}


\end{pspicture}
\end{center} 
\caption{Alice, Bob, and Calvin establish a common secret $\key$ in the presence of passive adversary Eve.}
\label{fig_connections}
\vspace{-0.3cm}
\end{figure*}

\subsection{An Example Agreement}

We will now illustrate the role of each step through a simple example (Figure~\ref{fig_connections}):
Alice wants to create a common secret with $2$ other nodes, Bob and Calvin;
both the Alice-Bob channel and the Alice-Calvin channel have the same erasure probability $\delta_1 = \delta_2 = \delta$;
the Alice-Eve channel has erasure probability $\delta_E$.
Assume that Eve is a passive adversary, i.e., she never performs any transmissions.

In step 1 of the initial phase, Alice transmits $N$ $x$-packets.  Assume that $N$ is large enough
that, at the end of this step, Bob has received $N_1 \rightarrow (1 - \delta) N$ $x$-packets,
Calvin has received $N_2 \rightarrow (1 - \delta) N$ $x$-packets,
and Bob and Calvin together have received
$N^* \rightarrow (1 - \delta^2) N$ $x$-packets.

In step 3 of the initial phase, Alice creates the $y$-packets, which encode all the secret information that
is shared, at this point, by Alice and each terminal separately:
She creates $M = \delta_E (1 - \delta^2) N$ $y$-packets.
Among these, there are $M_1 = \delta_E N_1$ $y$-packets, which are linear combinations of the $N_1$ $x$-packets received by Bob;
these encode all the information that is shared by Alice and Bob but not Eve.
Similarly, there are $M_2 \approx M_1$ $y$-packets, which are linear combinations of the $N_2 \approx N_1$ $x$-packets received by Calvin;
these encode all the information that is shared by Alice and Calvin but not Eve.
So, at the end of the initial phase, Alice shares $M_1$ secret packets with Bob and $M_1$ (different) secret packets with Calvin.

In the reconciliation phase, Alice first tries to reach a point where she shares with both Bob and Calvin all the $M$ $y$-packets.
To this end, in step 1 of the reconciliation phase, Alice creates the $z$-packets, which
encode the difference between what Bob and Calvin already know and what Alice wants them to know:
she creates and reliably broadcasts $M - M_1$ linear combinations of the $y$-packets.
In step 2 of the reconciliation phase,
Bob combines the $M_1$ $y$-packets he already knows with the $M - M_1$ linear combinations of the $y$-packets
broadcast by Alice and reconstructs all the $M$ $y$-packets (and Calvin does the same).

Now consider Eve:
Assume that $N$ is large enough that,
at the end of the initial phase, Eve has received $M_E \rightarrow (1 - \delta_E) N^*$ of the $x$-packets received by 
either Bob or Calvin or both.
Hence, Eve cannot reconstruct any of the $M = \delta_E N^*$ $y$-packets.
In step 1 of the reconciliation phase, Alice reliably broadcasts $M - M_1$ linear combinations of the $y$-packets 
(in order to fill in Bob's and Calvin's missing information).
Eve also receives this broadcast.
Hence, at the end of step 2 of the reconciliation phase, Eve can reconstruct $M - M_1$ of the $y$-packets.

Based on what we have said so far, at the end of step 2 of the reconciliation phase,
Bob and Calvin know all $M$ $y$-packets, while Eve knows $M - M_1$ $y$-packets.
Hence, in step 3 of the reconciliation phase, Alice creates the $s$-packets, 
which encode all the secret information that is shared, at this point, by Alice, Bob, and Calvin:
she creates $M_1$ $s$-packets, which are linear combinations of all the $M$ $y$-packets.
The concatenation of the $M_1$ $s$-packets is the common secret $\key$.

To recap, at the end of the initial phase, Alice shares $M_1$ different secret packets with each of Bob and Calvin,
whereas at the end of the reconciliation phase, she shares $M_1$ \emph{common} secret packets with both of them.
So, the reconciliation phase does not increase the amount of secret information shared by Alice and each terminal,
but ``redistributes'' information, such that Alice ends up sharing the same secret information with all the terminals.

\subsection{Discussion of Key Points}

The size of the established common secret is $L = \min_i M_i = \delta_E \min_i N_i = \delta_E (1 - \max_i \delta_i) N$:
the minimum number of $x$-packets that are received by a terminal but not Eve.
In other words, the size of the established common secret is determined by the \emph{weakest} terminal,
i.e., the one that shares the least amount of secret information with Alice at the end of the initial phase. 
This means that, if Alice can establish a secret of size $L$ with Bob, and then we add another terminal,
Calvin, who has better connectivity to Alice than Bob, then the three terminals can establish a common
secret of the same size $L$---i.e., adding Calvin to the group will not decrease the size of the established secret.

So, increasing the number of terminals from $2$ to $\numnodes$ does
not necessarily decrease the size of the established common secret.
For instance, if all the terminals have identical channels ($\delta_i
= \delta, \forall i$), then the size of the established common secret
is $L = \delta_E (1 - \delta) N$, which is equal to the size of the
secret established between two terminals in the example of
Section~\ref{sec:setup:main}.  
Moreover, the extra transmissions that Alice has to make in the reconciliation phase are 
$M - \delta_E (1 - \delta) N \le N - \delta_E (1 - \delta) N$, 
i.e., $M$ is upper-bounded by a constant that is independent of $\numnodes$.
As we will see in the analysis
section, this independence from $\numnodes$ is key to the scalability
of our protocol.

Linear coding is used to two different effects by our protocol:
(1) As a means to fill in the information missing from each terminal by transmitting the minimum number of packets:
In step 1 of the reconciliation phase, the $M - \min_i M_i$ $z$-packets are linear combinations of the $M$ $y$-packets;
given that each terminal already knows at least $\min_i M_i$ $y$-packets from the initial phase,
it can reconstruct all $M$ $y$-packets.
(2) As a means to perform privacy amplification, i.e., combine all the information known to the terminals but not to Eve:
In step 3 of the initial phase, Alice creates $\delta_E N^*$ $y$-packets that are linear combinations of the 
$N^*$ $x$-packets received by at least one terminal;
assuming Eve has missed $\delta_E N^*$ of these $x$-packets, she cannot reconstruct any of the $y$-packets.
Similarly, in step 3 of the reconciliation phase, Alice creates $\min_i M_i$ $s$-packets that are linear combinations 
of the $M$ $y$-packets known to all the terminals; assuming Eve has missed $\min_i M_i$ $y$-packets, 
she cannot reconstruct any of the $s$-packets.

Point (2) assumes that $N$ (hence also $N^*$) is large enough that, if the Alice-Eve channel has erasure probability $\delta_E$,
then Eve misses close to $\delta_E N^*$ of the $N^*$ $x$-packets.
Of course, it is theoretically possible that Eve gets lucky and receives significantly more $x$-packets than expected;
however, by picking the right value for $N$, we can make this event arbitrarily unlikely
(e.g., as likely as Eve correctly guessing the value of the secret $\key$ by randomly picking a number between $0$ and $|\key|$).
In Section~\ref{sec:jam}, where we present our experimental results, we show exactly how much information
Eve collects about every generated common secret $\key$.

\section{Protocol Analysis}
\label{sec:analysis}

In this section, we state certain properties of the Basic secret-agreement protocol and also 
present a formal argument on why we chose this particular protocol over a more obvious alternative. 

\begin{lemma}
\label{lem:secrecy}
If the theoretical network conditions hold, there exists a sufficiently large $N$ for which
the Basic secret-agreement protocol is information-theoretically secure against a passive adversary.
\end{lemma}

\begin{lemma}
\label{lem:efficiency}
If the theoretical network conditions hold, there exists a sufficiently large $N$ for which
the Basic secret-agreement protocol achieves efficiency
$$\eff = \frac{\delta_E (1-\delta)}{1+\delta_E(\delta-\Delta)},$$
where
$\delta = \max_i\{\delta_i\}$ 
and $\Delta = \delta_1 \delta_2 \ldots \delta_{\numnodes-1}$.
\end{lemma}

To give a sense of the achieved efficiency, we consider the case where $\delta_i = \delta_E$
(all the channels between Alice and any node are identical) and plot, in Figure~\ref{fig:SecRateComp} (solid lines), the efficiency 
of our protocol as a function of the erasure probability of the channels, for different values of the number of terminals $\numnodes$.
The bell-shape of the curve is explained as follows: when the erasure probability is $0$,
Eve misses none of the packets transmitted by Alice, which means that Alice cannot establish any secret
with the other terminals, no matter how many packets she transmits---hence, the efficiency of the protocol is $0$;
when the erasure probability is $1$, Eve misses all of the packets transmitted by Alice,
but so do all the other terminals, so, again, the efficiency of the protocol is $0$;
the maximum efficiency is achieved somewhere in between (at erasure probability $0.5$ when we have $\numnodes = 2$ terminals,
and at $\sqrt{2}-1$ when $\numnodes\rightarrow\infty$).
The shift of the maximum point is due to the additional (by at most $N - \min_i M_i$) transmissions performed in the reconciliation phase.
Note that, as the number of terminals goes to infinity, the maximum efficiency of our protocol approaches $20$\%---a substantial 
non-zero value.

\begin{figure}[!t]
\begin{center}
 \includegraphics[width=3.5in,height=1.6in]{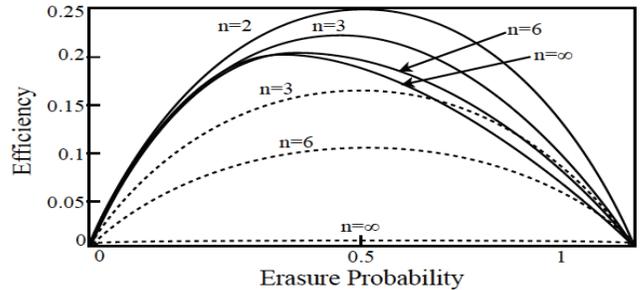}
\end{center}
\vspace{-1em}
\caption{
Efficiency of secret agreement as a function of the erasure probability of the channels
(assuming identical erasure channels) for our protocol (solid lines) and the alternative protocol (dashed lines).} 
\label{fig:SecRateComp}
\vspace{-0.5cm}
\end{figure} 

The Basic protocol scales well with the number of terminals because we try to leverage broadcasting as much as possible.
If we were, instead, to require pairwise secret establishment between Alice and each terminal,
efficiency would quickly go to $0$ with the number of terminals.
To see this, consider the following, conceptually simpler alternative to the Basic protocol:
Alice establishes a separate secret $\key_i$ with each other terminal $T_i$ (using the initial phase of the Basic protocol)
and uses $\key_i$ to convey a common secret $\key$ to each $T_i$
(e.g., reliability broadcasts $\mathcal{M}_i = \key \oplus \key_i$ for all $i \neq 0$).
Its efficiency is
$\eff^{\text{(alt)}} = \frac{\delta_E (1-\delta) }{1+(\numnodes-2)\delta_E (1-\delta) },$
where $\delta = \max_i \delta_i$.
Figure~\ref{fig:SecRateComp} shows this efficiency (dashed lines),
plotted against the efficiency of the Basic protocol (solid lines), as a function of the erasure probability of the channels,
assuming all channels are identical.
Notice that, unlike the efficiency of our protocol,
$\eff^{\text{(alt)}}$ quickly goes to zero as the number of terminals $\numnodes$ increases.

\begin{lemma}
\label{lem:optimality}
If the theoretical network conditions hold,
then, for $\numnodes = 2$ terminals, the Basic protocol achieves maximum efficiency.
\end{lemma}
This is directly derived from Lemma~\ref{lem:efficiency}:
for $\numnodes =2$ terminals, we achieve efficiency $E = \delta_E (1 - \delta_1)$,
which is the maximum possible \cite{Ma-IT93}.

\begin{lemma}
\label{lem:complexity}
Each terminal that participates in the Basic protocol  
executes an algorithm that is polynomial in $N$.
\end{lemma}

The proofs of Lemmas~\ref{lem:secrecy} and~\ref{lem:complexity} are in the Appendix, Section~\ref{app:1}.
We omit the proof of Lemma~\ref{lem:efficiency}, which is straightforward.

\section{Authentication}
\label{sec:authenticate}


The Basic protocol is information-theoretically secure against passive adversaries,
but is vulnerable to active attacks: what if Eve impersonates a terminal, participates in the protocol,
and learns the common secret? 

To protect against impersonation attacks, the (true) terminals need to share an
initial common secret $\initkey$, of sufficient size to authenticate each other until 
they successfully complete one round of the Basic protocol;
once they have successfully completed a round
(and generated a new common secret $\key$ with $|\key| \gg |\initkey|$),
they can use a part of $\key$ to authenticate each other until the next successful round completion.

Since we are aiming for information-theoretic security,
the terminals use an unconditionally secure authentication code~\cite{stinson}.
Such a code provides a function $\auth(\mu, \initkey)$, which returns an authenticator $\alpha$ for message $\mu$ given key $\initkey$,
such that: an entity that does not know $\initkey$ can
generate a valid message/authenticator pair (launch a successful impersonation attack)
with probability $\frac{1}{|\mathcal{A}|}$, where $|\mathcal{A}|$ is the size of the authenticator space.

Choosing at which step(s) of the Basic protocol to perform the authentication involves a trade-off between efficiency
and the adversary model that we want to consider:
At one extreme, each terminal appends an authentication code to every single packet it transmits or reliability broadcasts.
At the other extreme, the terminals authenticate each other only at the last step of the reconciliation phase
(i.e., each terminal obtains proof that all the other terminals with which it has created a common secret know the initial
common secret $\initkey$).
The former requires a significantly larger $\initkey$ to provide information-theoretic guarantees (because it reveals many more
authenticators to the adversary).
The latter makes the protocol vulnerable to a simple denial-of-service attack: in every protocol round, Eve impersonates a terminal
and learns the common secret; at the end of the reconciliation phase, she fails to authenticate herself to the (true) terminals,
causing the common secret to be discarded and the protocol to restart.

We chose a solution in the middle, which, in our opinion, offers a good balance:
authentication happens at the end of the initial phase, after step $4$:

\begin{itemize}
\item[4-A]
Alice performs \emph{pair-wise authentication} (described below) with each terminal $T_{i \neq 0}$,
using the $y$-packets that $T_i$ has reconstructed and the initial common secret $\initkey$.
If Alice fails to authenticate herself to $T_i$, $T_i$ stops participating in the protocol.
If $T_i$ fails to authenticate itself to Alice, Alice excludes $T_i$ from the agreement and
discards all the $y$-packets that $T_i$ has reconstructed.
\end{itemize}
Pair-wise authentication consists of the following steps:
\begin{enumerate}
\item
Alice concatenates $W$ bits selected from the $y$-packets that $T_i$ has reconstructed in the initial phase, step 4,
creates a message $\mu_i$ that contains this concatenation,
and reliably broadcasts $\alpha_i = \auth(\mu_i, \initkey)$.
\item
$T_i$ creates a message $\mu_i'$ in the same way and checks whether $\alpha_i = \auth(\mu_i', \initkey)$.
If yes, Alice has successfully authenticated herself to $T_i$, otherwise, she has failed.
\item
$T_i$ concatenates a different set of $W$ bits from the $y$-packets that it has reconstructed in the initial phase, step 4,
creates a message $\nu_i$ that contains this concatenation,
and reliably broadcasts $\beta_i = \auth(\nu_i, \initkey)$.
\item
Alice creates a message $\nu_i'$ in the same way and checks whether $\beta_i = \auth(\nu_i', \initkey)$.
If yes, $T_i$ has successfully authenticated itself to Alice, otherwise, it has failed.
\end{enumerate}

Regarding the size of the relevant parameters: 
Since we are using an unconditionally secure authentication code, 
Eve can launch a successful impersonation attack with probability $\frac{1}{|\initkey|}$;
for $|\initkey| = 32$ bits, this becomes $0.232 \cdot 10^{-9}$.
We authenticate messages of size $W$;
unconditionally secure authentication codes require a key of twice the size of the authenticated message, 
hence $W = |\initkey|/2 = 16$ bits.

Now suppose that Eve is an active adversary.
First, she impersonates a terminal other than Alice.
In this case, Eve fails to authenticate herself to Alice (initial phase, step $4$-A) because she does not know $\initkey$,
causing Alice to exclude her from the agreement.
This means that Alice discards all the $y$-packets that Eve has reconstructed, hence,
at the end of the initial phase, Eve cannot reconstruct any of the (remaining) $y$-packets,
i.e., she is in the same position as a passive adversary.
Second, suppose that Eve impersonates Alice.
In this case, Eve fails to authenticate herself to any of the other terminals (initial phase, step $4$-A) because she
does not know $\initkey$, causing the terminals to stop talking to her.

\section{Adapting to a Real Network}
\label{sec:real}

In this section, we describe how we adapt our secret-agreement protocol 
to a small wireless testbed ($14$ m$^2$), where the theoretical network conditions do not hold.
Instead, we use interferers to introduce noise and ensure that an adversary does not receive all 
the packets received by any terminal, as long as she has a minimum physical distance ($1.76$ m)
from each terminal.

\subsection{Setup and First Try}
\label{sec:real:setup}

Our testbed (Figure~\ref{fig:testbed}) covers a square area of $14$
m$^2$.  We use $\numintf = 6$ interferers, which are WARP nodes
\cite{WARPref}, each with two directional antennae, each with a narrow
$3$-dB $22$-degree beam.  We deploy up to $\numnodes = 8$ terminals
and one adversary, Eve, which are Asus WL-500gP wireless routers
running $802.11$g (at $2.472$ GHz, transmit power $3$ dBm) in ad-hoc
mode.  During our experiments, when a terminal transmits, it sends $100$-byte packets
at a rate of $1$ Mbps.  

We logically divide the testbed area in $3$ rows and $3$ columns of equal width,
place Eve in one of the $9$ logical cells, and the terminals in
various positions around her, but not in the same cell. Our rationale
is the following: if a group of wireless nodes want to exchange a
secret, it is reasonable to require from each of them to stand at
least some minimum distance away from any other wireless node. In our
testbed, this minimum distance is $1.76$ m (the diagonal of a logical
cell), and, as we explain below, it was determined by the shape of the
interferers' beams; with a narrower beam, we would have achieved a
smaller minimum distance.

We place the interfering antennae along the perimeter of the covered area ($3$ on each side);
we turn them on and off, such that, at any point in time, one pair of antennae creates noise along a row, 
while another pair creates noise along a column; since we have $9$ row/column
combinations, there are $9$ different noise patterns.  
To choose the width of the rows/columns, we performed a simple calibration using only two interfering antennae:
we placed the antennae at opposite ends of the first row;
we chose the width of the row to be the maximum width such that, when the antennae were on,
any wireless node located in the row did not receive any other signal transmitted from within the $14$ m$^2$ covered area.
Hence, the number of rows and columns was
determined by the shape of the interfering antennae's beams; a narrower beam
would have resulted in more rows and columns given the same area.

Each experiment is divided in time slots; at the beginning of each time slot, we turn on different interferers,
such that, by the end of the experiment, we have rotated through all $9$ noise patterns.

\begin{figure}[t!]
\begin{center}
\psset{unit=0.028in}
\begin{pspicture}(5,-2)(87,60)
\psset{linewidth=0.2mm}

\psline[linecolor=black,linestyle=solid,linewidth=0.2mm]{-}(20,0)(20,60)
\psline[linecolor=black,linestyle=solid,linewidth=0.2mm]{-}(40,0)(40,60)
\psline[linecolor=black,linestyle=solid,linewidth=0.2mm]{-}(60,0)(60,60)
\psline[linecolor=black,linestyle=solid,linewidth=0.2mm]{-}(80,0)(80,60)

\psline[linecolor=black,linestyle=solid,linewidth=0.2mm]{-}(20,0)(80,0)
\psline[linecolor=black,linestyle=solid,linewidth=0.2mm]{-}(20,20)(80,20)
\psline[linecolor=black,linestyle=solid,linewidth=0.2mm]{-}(20,40)(80,40)
\psline[linecolor=black,linestyle=solid,linewidth=0.2mm]{-}(20,60)(80,60)

\psframe*[linecolor=red](48,30)(52,34)\rput(50,27){\bf Eve}

\cnode[fillcolor=blue, fillstyle=solid](30,28){1.5}{T_1}
\nput[labelsep=1]{180}{T_1}{\large $T_2$}

\cnode[fillcolor=blue, fillstyle=solid](33,51){1.5}{T_2}
\nput[labelsep=1]{180}{T_2}{\large $T_3$}

\cnode[fillcolor=blue, fillstyle=solid](48,48){1.5}{T_3}
\nput[labelsep=1]{0}{T_3}{\large $T_4$}


\cnode[fillcolor=blue, fillstyle=solid](70,50){1.5}{T_4}
\nput[labelsep=1]{0}{T_4}{\large $T_5$}

\cnode[fillcolor=blue, fillstyle=solid](68,10){1.5}{T_5}
\nput[labelsep=1]{0}{T_5}{\large $T_0$}

\cnode[fillcolor=blue, fillstyle=solid](31,10){1.5}{T_6}
\nput[labelsep=1]{180}{T_6}{\large $T_1$}



\pnode(15,40){i_2}
\pnode(15,60){i_3}
\ncline[linewidth=0.15mm]{<->}{i_2}{i_3}\lput*{:U}{1.25m}

\pnode(40,0){i_4}
\pnode(60,20){i_5}
\ncline[linewidth=0.15mm]{<->}{i_4}{i_5}\lput*{:U}{1.76m}

\end{pspicture}
\caption{\label{fig:testbed} A testbed configuration. The six round nodes ($T_i$) are trusted terminals
that are trying to establish a common secret $\key$. The square node (Eve) is an adversary whose goal is to
guess $\key$. }
\end{center}
\vspace{-0.7cm}
\end{figure}
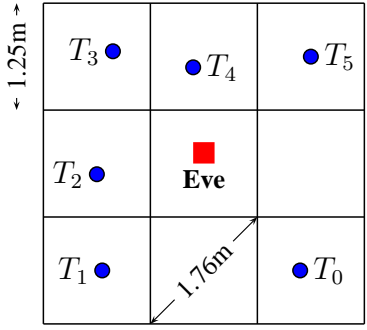

The goal of rotating noise patterns is to emulate independent
packet-erasure channels: Suppose Alice transmits $N$ $x$-packets
during each experiment, equally spread across the $9$ time slots.
Assume that, during each time slot, each node located in the row or
the column that are interfered with receives $0$ of Alice's packets,
while each of the remaining nodes receives all of Alice's packets.
Now consider the positioning of the terminals shown in
Figure~\ref{fig:testbed}.  Because terminal $T_1$ is on the same row
with Eve, $2$ out of the $9$ noise patterns affect Eve but not $T_1$;
hence, if, during an experiment, Alice transmits $N$ packets, $0.22 N$
of these packets are received by $T_1$ but not Eve.  
This is the same outcome that we would have,
if the Alice-$T_i$ and Alice-Eve channels were independent, with
$\delta_E (1 -\delta_1) = 0.22$.

\vspace{0.2cm}
\noindent
{\bf Our First Try.}
The reason why we need the theoretical network conditions in the Basic protocol is that they
guarantee two facts: 
(i) each terminal $T_i$ receives $M_i = \delta_E N_i$ $x$-packets that are not
received by Eve, and
(ii) the union of all terminals receives $M =
\delta_E N^*$ $x$-packets that are not received by Eve.
These two numbers determine the main parameters of
the protocol, namely how many $y$-packets ($M$), $z$-packets ($M -
\min{M_i}$), and $s$-packets ($L = \min{M_i}$) Alice creates.  

In our
testbed, we expected our controlled interference to ensure two similar facts: 
(i) each terminal $T_i$ receives at least $\min{M_i} = 0.22
N$ $x$-packets that are not received by Eve, and
(ii) the union of all terminals receives at least $M$
$x$-packets that are not received by Eve, where $M$ can be computed based on the
number of cells occupied by terminals (we skip the computation for lack of space).
Hence, even though the
theoretical network conditions do not hold, we still know how many $y$-,
$z$- and $s$-packets Alice should create.

However, this rationale assumes that, when Alice transmits, all the
nodes that are located in an interfered-with row or column receive
nothing, while the rest receive everything.  
It turned out that, in practice, this cannot be guaranteed with probability $1$, due to random channel-propagation effects,
which creates two new problems:
(1) {\bf Zero secret size}: in rare, but statistically significant occasions, during an experiment, Eve receives
all the $x$-packets received by one of the terminals, i.e., $\min{M_i} = 0$, 
which means that the terminals cannot create any common secret at all.  
(2) {\bf Unpredictable secret size}: the value of $\min{M_i}$ varies significantly between experiments,
which means that the terminals do not know how big a secret they should create (such that Eve has $0$ information about it).
To address these two problems, we adapt our secret-agreement protocol as described in the next three sections.

\subsection{Gist}

To address the ``zero secret size'' problem,
we make all the terminals take turns in transmitting $x$-packets.
The idea is to make each terminal $T_i$ receive information through multiple different channels (as opposed to 
receiving information only from Alice),
making it unlikely that Eve will collect the same information with $T_i$.
In particular, Eve collects the same information with terminal $T_i$ only when:
for {\em every} single terminal $T_{j\neq i}$, the channel between $T_j$ and $T_i$
happens to be the same with the channel between $T_j$
and Eve, throughout the experiment.  This never happened in any of the
experiments that we ran.  

To address the ``unpredictable secret size'' problem, we estimate the amount of information that 
Eve collects \emph{based on the information collected by the terminals}.
I.e., we essentially pretend that each terminal $T_i$ is Eve and that all the other terminals want to establish a common
secret that is unknown to $T_i$; since we know which packets were received by each terminal (including $T_i$),
we can compute exactly what the size of this supposed secret should be.
Then we combine all these computations to estimate the size of the actual common secret.

We made this last choice based on the following observations:
Channel behavior varies significantly over time, to the point where we cannot estimate or even upper-bound how much information
Eve collects during one experiment based on how much information she collected during past experiments.
Channel behavior also varies over \emph{space}, but less so: 
if, during an experiment, node $T_i$ receives many packets in common with neighbor $T_j$,
then node $T_i$ most likely receives many packets in common with its other neighbors as well.
It turns out that, by measuring how many packets each pair of neighboring terminals receive in common during one experiment, 
we can estimate quite accurately how many packets any terminal and Eve receive in common \emph{in the same experiment}.

\subsection{Adapted Protocol Description}
\label{subsec:AdaptProtocol}

\paragraph*{Initial Phase}
For $j = 1..\numnodes$:
\begin{enumerate}
\item
Terminal $T_j$ transmits $N$ packets (we will call them \emph{$x$-packets}).
\item
Each terminal $T_{i\neq j}$ reliably broadcasts a feedback message specifying which $x$-packets it received.
\item
$T_j$ does the following:
\begin{enumerate}
\item
It counts the number of $x$-packets $N_i$ received by terminal $T_i$, for all $i \neq j$.
\item
It counts the number of $x$-packets $M_{ik}$ received by both $T_{i}$ and $T_{k}$, for all $i, k \neq j$.
\item
It computes $\tilde{\delta_i} = \frac{\max_k{M_{ik}}}{N_i}$ for all $i \neq j$.
\item
It computes $\deltaE = \max_i{\tilde{\delta_i}}$.
\item
It performs step 3 of the initial phase of the Basic protocol, using $\delta_E = \deltaE$.
\end{enumerate}
As a result, $T_j$ creates $M^j = \tilde{\delta_E} N$ linear combinations of the $x$-packets
(we will call them \emph{$y$-packets}) and reliably broadcasts the coefficients it used to create them.
\item
Each terminal $T_{i \neq j}$ reconstructs as many (say $M_i$) of the $M^j$ $y$-packets as it can (based on the $x$-packets it received in step $1$).
\end{enumerate}

\smallskip
\paragraph*{Reconciliation Phase}
\begin{enumerate}
\item
Alice identifies which $M_i$ $y$-packets are known to each terminal $T_i$.
She provides them as input to the program specified in the Appendix, Section~\ref{sec:cooperative},
which outputs a non-negative integer $K_i$ for all $i$.
Then she reliably broadcasts all the $K_i$.
\item
Each terminal $T_i$ creates $K_i$ linear combinations of the $M_i$ $y$-packets that it reconstructed in the initial phase
(we will call them \emph{$z$-packets}).
It reliably broadcasts both the contents and the coefficients of the $z$-packets.
\item
Each terminal $T_i$ combines the $z$-packets it received with the $y$-packets it reconstructed in the initial phase,
and reconstructs all the $M = \sum_j{M^j}$ $y$-packets.
\item
Alice creates $L = M - \sum_i{K_i}$ linear combinations of the $M$ $y$-packets
(we will call them \emph{$s$-packets}), using the construction specified in Lemma~\ref{lem:linear_3} (Section~\ref{app:1}, Appendix).
She reliably  broadcasts the coefficients she used to create all the $s$-packets.
\item
Each terminal $T_i$ reconstructs all the $s$-packets. The common secret is their concatenation
$\key = \langle s_1, \ldots, s_L \rangle$.
\end{enumerate}

\subsection{Discussion of Key Points}

We illustrate the key points of the protocol by considering again the example
where Alice, Bob, and Calvin want to establish a common secret $\key$ in the presence
of passive eavesdropper Eve.

The first difference from the Basic protocol is that we do not know how many $x$-packets
Eve receives in common with each terminal, so, we estimate it based on how many $x$-packets
various pairs of terminals receive in common (initial phase, step 3).
For instance, suppose that, in step 1 of the initial phase, Alice transmits $N=10$ $x$-packets,
Bob receives $x_1$, $x_2$, $x_3$, $x_4$, $x_7$, and $x_8$, while Calvin receives $x_1$, $x_3$, $x_5$, and $x_6$.
Hence, Bob receives $N_1 = 6$ $x$-packets, Calvin receives $N_2 = 4$ $x$-packets,
Bob and Calvin together receive a total of $N^* = 8$ $x$-packets, while they receive
$M_{12} = 2$ $x$-packets in common.
Hence, in step 3 of the initial phase, Alice computes
$\tilde{\delta_1} = 0.33$, $\tilde{\delta_2} = 0.5$, and $\deltaE = 0.5$,
and she creates a total of $M^0 = \deltaE N^* = 4$ $y$-packets.

What we lose relative to the Basic protocol is that we cannot guarantee a minimum reliability,
because we do not know how much information Eve collects during the initial phase:
It is possible that Eve receives more $x$-packets in common with the terminals than we estimate,
which means that, at the end of the initial phase, she knows some fraction of the $M$ $y$-packets,
hence, at the end of the reconciliation phase, she knows some information about the $L$ $s$-packets.
Note that this does \emph{not} mean that Eve knows the common secret $\key$, only that she knows
some information about it, which increases her probability of guessing it right.

A side-effect of estimating the amount of information collected by Eve based on the information collected by the terminals
is that the performance of the protocol depends on the number of terminals $\numnodes$: the more terminals we have,
the more we learn about the quality of channels throughout the network,
hence we can estimate the quality of Eve's channels better.
For instance, if we have only $\numnodes = 3$ terminals, when Alice transmits,
she estimates how many $x$-packets were received in common by Bob and Eve based on how many
$x$-packets were received in common by Bob and Calvin; if it happens that the channel from Alice to Eve
is significantly different from the channel from Alice to Bob, then the estimate is inaccurate.
For $\numnodes = 2$ terminals, the protocol does not work at all, because Alice has no way of estimating
how many $x$-packets Bob received in common with Eve.

The second difference from the Basic protocol is that the terminals take turns in
transmitting $x$-packets (initial phase, step 1) and creating $y$-packets (initial phase, step 3).
Hence, at the end of the initial phase, there exists no terminal that knows all the $y$-packets, and
the reconciliation needs to happen in a distributed manner:
we need to solve a program that takes as input which terminal knows which $y$-packets
and outputs how many $z$-packets (linear combinations of $y$-packets) each terminal needs to reliably broadcast,
such that all terminals learn all the $y$-packets (reconciliation phase, steps 1 and 2).

What we lose relative to the Basic protocol is that we cannot guarantee a minimum efficiency,
because we do not know how much information the terminals will have to broadcast during the reconciliation phase:
At the end of the initial phase, all the terminals together have created $M$ $y$-packets, where $M$ depends on network conditions.
In the reconciliation phase, all the terminals together broadcast $K$ $z$-packets (linear combinations
of $y$-packets), where $K$ also depends on network conditions.
Hence, at the end of the reconciliation phase, Eve knows at least $K$ of the $M$ $y$-packets,
and the terminals create $L = M - K$ $s$-packets (linear combinations of the $y$-packets).
This means that the efficiency of the protocol is $\frac{M - K}{\numnodes N + M - K}$,
which depends on the network conditions during the experiment.

In summary, once we do not assume perfect knowledge of network conditions,
we cannot offer formal guarantees about the reliability and efficiency of a protocol that
relies precisely on these network conditions to generate a secret;
we have to assess its reliability and efficiency experimentally, 
for the particular space where we want to deploy it.

\section{Experimental Evaluation}
\label{sec:jam}

In this section, we experimentally evaluate the Adapted secret-agreement protocol in our testbed.

When we refer to an ``experiment,'' we mean that
we place $\numnodes$ terminals and Eve on the area covered by our testbed, such that each cell is occupied by at most one node,
and we run one round of the Adapted protocol.
We run one such experiment for each possible positioning of $\numnodes$ terminals and Eve,
and we run one such set of experiments for $\numnodes = 3$ to $8$ terminals.
For instance, we run $504$ experiments with $\numnodes = 3$ terminals (because there are $504$ different
ways to position $3$ terminals and Eve on our testbed),
while we run $9$ experiment with $\numnodes = 8$ terminals (because there are $9$ different ways
to position $8$ terminals and Eve on our testbed).

\subsection{Efficiency and Reliability of the Adapted Protocol}

Each graph we present shows efficiency or reliability as a function of the number of terminals $\numnodes$;
for each value of $\numnodes$, we show three values: ``minimum'' is the minimum efficiency/reliability achieved
during any experiment with $\numnodes$ terminals; ``average'' is the average of the efficiency/reliability achieved
across all experiments with $\numnodes$ terminals; ``$50$th quartile'' and ``$95$th quartile'' is the minimum 
efficiency/reliability achieved during $50$\% and $95$\% of the experiments with $\numnodes$ terminals.

Figure~\ref{fig:eff} shows the efficiency of the Adapted protocol:
For $\numnodes = 8$ terminals, it has minimum efficiency $E_{min} = 0.038$; given that the terminals transmit at rate $1$ Mbps, 
this efficiency yields $38$ secret Kbps.
For $\numnodes = 6$ terminals, $E_{min} = 0.028$, which yields $28$ secret Kbps.
The reason efficiency decreases with the number of terminals is related to the ``zero secret size'' problem
(Section~\ref{sec:real}):
when Alice transmits, it is possible that Eve receive all the $x$-packets received by Bob
(if the Alice-Bob channel happens to be the same with the Alice-Eve channel);
we side-stepped this problem by making all the terminals transmit $x$-packets, thereby creating more channel diversity;
however, the fewer the terminals we have, the less the diversity we create, hence the more likely
it is for Eve to receive a large fraction of the $x$-packets received by Bob (or any one terminal).

Figure~\ref{fig:rel} shows the reliability of the Adapted protocol:
For $\numnodes = 8$ terminals, it has minimum reliability $R_{min} = 1$,
i.e., in any experiment, Eve can correctly guess the value of a secret bit with probability $2^{-1} = 0.5$
and the value of an $s$-packet with probability $2^{-800} \rightarrow 0$
(each packet is $800$ bits).
For $\numnodes = 6$ terminals, $R_{min} = 0.2$, i.e., in the worst case,
Eve can correctly guess the value of a secret bit with probability $2^{-0.2} = 0.87$,
but the value of an $s$-packet still with probability $2^{-0.2\cdot 800} \rightarrow 0$.
The reason reliability decreases with the number of terminals is related to the ``unpredictable secret size''
problem (Section~\ref{sec:real}):
each terminal estimates how many $y$-packets to create based on information provided by the other terminals;
the fewer the terminals, the less accurate the estimate, hence the more likely it is to create more $y$-packets
than are secret to Eve.
Note that the ``bad'' experiments (where we achieve very
low reliability) are relatively few: for $\numnodes = 6$ terminals, in $95$\% of the experiments (i.e., possible node placements),
$R_{min} = 0.5$, i.e., Eve can correctly guess the value of a secret bit with probability
$0.7$ and an $s$-packet with probability $2^{-0.5\cdot 800} \rightarrow 0$. 
Also, for any number of terminals, in at least half of the node placements, we achieve minimum reliability $1$
(the $50$\% quartile in Fig. \ref{fig:rel} is always $1$).

\begin{figure}[!t]
\subfigure[\label{fig:eff}Efficiency]{
\includegraphics[width=3.5in]{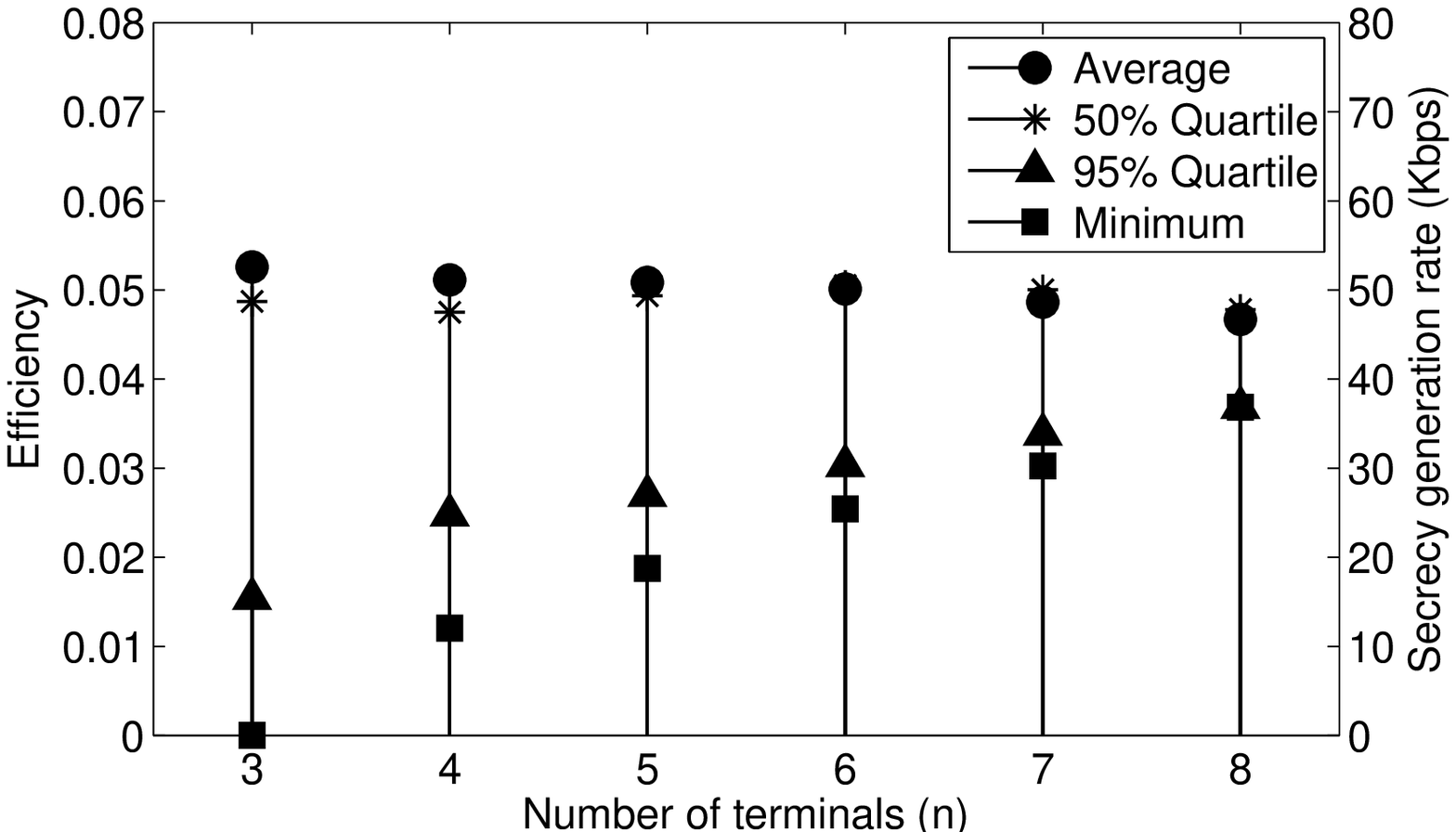}}
\subfigure[\label{fig:rel}Reliability]{
\includegraphics[width=3.5in]{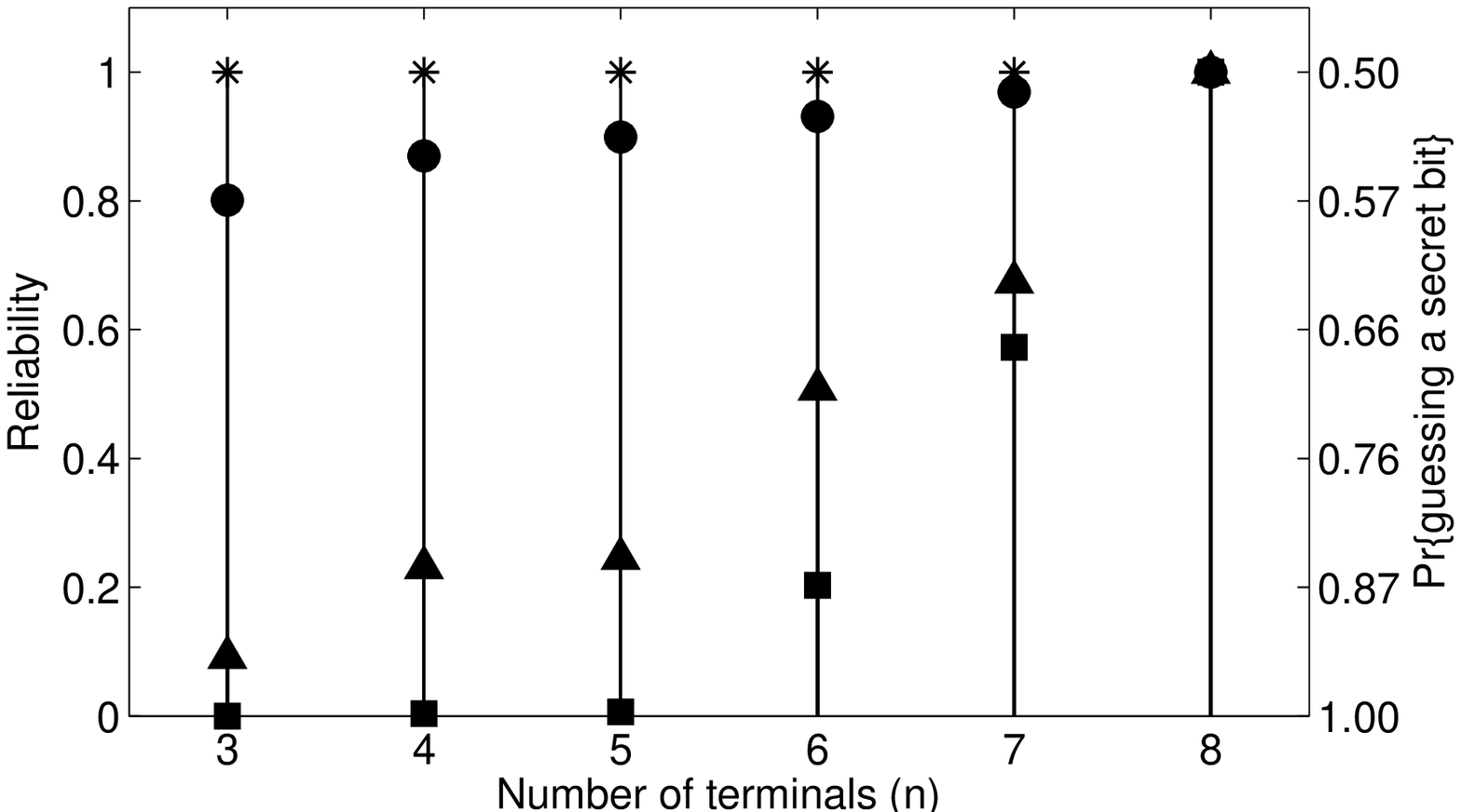}}
\caption{\label{fig:adapted} Performance of the Adapted protocol as a function of the number of terminals $\numnodes$.
For $\numnodes = 8$ terminals, we generate at least $38$ secret Kbps with minimum reliability $1$.
For $\numnodes = 7$ and $6$ terminals, we generate at least $28$ secret Kbps with minimum reliability below $1$,
but such that Eve still has negligible probability of correctly guessing an $s$-packet.
For fewer terminals, Eve's probability of correctly guessing an $s$-packet becomes non-negligible.}
\vspace{-0.5cm}
\end{figure}

\subsection{Privacy Amplification}

In the Adapted protocol, each terminal generates as many $y$-packets as (it estimates to be) possible,
such that these packets remain secret from Eve;
as we saw in the last section, for fewer than $8$ terminals,
this results in creating a larger secret than appropriate, hence achieving reliability well below $1$.

One practical---if not elegant---way of increasing reliability would be to add a conservative privacy-amplification step
at the end of the reconciliation phase of the Adapted protocol: instead of creating $L = M - \sum_i K_i$
linear combinations of the $y$-packets, create $\alpha L$ combinations, where $\alpha \in (0,1)$ depends 
on the target secret bitstream rate that we want to generate.
For instance, with $\numnodes = 6$ terminals, the Adapted protocol yields a $28$ Kbps bitstream,
but achieves minimum reliability $0.2$, because it creates a larger secret than it should;
if we only care to generate a $1$ Kbps secret bitstream, then we can instruct the protocol to create
$0.035 L$ $s$-packets at the end of the reconciliation phase, which would result in higher reliability.

\begin{figure}[!t]
\includegraphics[width=3.5in]{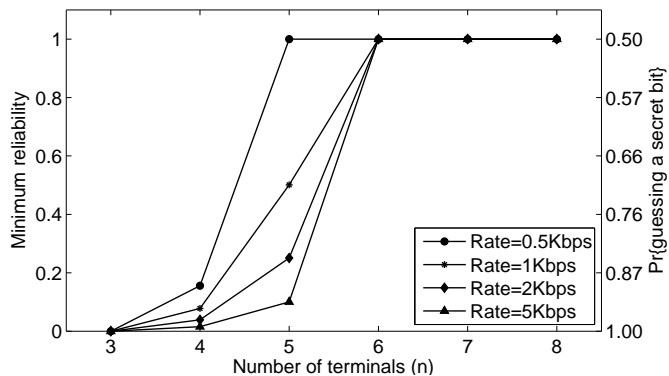}
\caption{\label{fig:tradeoff} 
Performance of the Adapted protocol when using privacy amplification.
The plot shows minimum reliability as a function of the number of terminals $\numnodes$, for different target secret bitstream rates.
For $\numnodes = 6$ terminals, we get minimum reliability $1$ for $5$ secret Kbps or less.
For $\numnodes = 5$ terminals, we get minimum reliability $1$ for $0.5$ secret Kbps or less.
}
\vspace{-0.5cm}
\end{figure}

Figure~\ref{fig:tradeoff} shows the minimum reliability that we get as a function of the number of terminals $\numnodes$,
for different target secret bitstream rates:
if we instruct our protocol to generate $5$ secret Kbps, 
we get minimum reliability $R_{min} = 1$ for $\numnodes = 6$ or more terminals;
if we instruct our protocol to generate $0.5$ secret Kbps, 
we get minimum reliability $R_{min} = 1$ for $\numnodes = 5$ or more terminals.
This suggests that our protocol could achieve higher reliability by adapting its secret-generation rate according 
to the number of participating terminals; exploring how exactly to perform this adaptation is part of our future work.

\section{Discussion}
\label{sec:discuss}

In this section, we discuss the limitations of our proposal, outline ideas on how to address them,
and state open challenges.

{\bf Collusion.}
We can model collusion by assuming that Eve is physically present in many network locations at the same time;
in that case, she can use the union of the packets received at all locations to compromise the security of the common secret.
In the theoretical model, this scenario is straightforward to address: by being present in many locations, Eve essentially decreases
her erasure probability $\delta_E$;
hence, the Basic protocol does not need to change, it only needs to use Eve's effective erasure probability.
In a real testbed, we will have to be more conservative about the number of $y$-packets (hence, common secret) we create,
depending on the level of collusion we want to prevent.
For instance, suppose that Eve may be present in up to two network locations;
Alice will have to estimate how many packets Bob and Eve received in common based on
how many packets Bob, Calvin, and some other terminal received in common
(as opposed to considering only Bob and Calvin).
We expect the cost to be  reduced efficiency.


{\bf Less controlled environment.}
Our testbed implies a controlled environment, e.g., a military facility,
where (i) we can rely on trusted interferers to create artificial noise, 
and (ii) we can assume that nodes located outside the perimeter of the testbed
cannot overhear transmissions initiated from within the perimeter
(e.g., because the testbed is within walls that are resistive to outgoing radio signals).
Our ultimate goal is not a such controlled environment.
We plan to explore the idea of having the terminals themselves generate the necessary noise,
such that we do not require trusted interferers or well defined boundaries.

{\bf Eve has a custom physical layer.}
In our testbed, Eve possesses a standard physical layer, which means
that she cannot recover packets that are discarded by her physical layer.
One might argue that, if she had access to a custom physical layer, she could buffer these packets
and collect some amount of information from them (because, presumably, not all the bits of a discarded packet
are received incorrectly).
Even though we did not perform experiments with custom hardware, we have some evidence that our results
would still hold: we measured the bit error rate experienced by a terminal located in a row or column
that is interfered with, for a wide range of packet sizes; we found that the bit error rate remains constant,
independently from the packet size; this is consistent with our expectation that, when a packet
is lost due to our artificial interference, most of its bits are received incorrectly by the physical layer of
the device.

{\bf Eve has directional antennae.}
In our testbed, Eve is a commodity wireless node. One might argue that, if she had access to
larger and/or directional antennae, she could use them to cancel out the interferes' noise.
Making our interference robust to such attacks is something we want to explore in our future work.
However, we should note that our interference is nearly omnidirectional, hence we expect larger and/or directional
antennae not to make it significantly easier for Eve to cancel it out.

{\bf Residual uncertainty.}
In our testbed, we achieve reliability below $1$ in certain experiments, 
which is a result of the ``unpredictable secret size'' problem (Section~\ref{sec:real}).
We want to improve this through more sophisticated interference, e.g., by using beam forming.
However, we would like to note that some level of uncertainty will be unavoidable in a real network.
This is similar to the uncertainty that exists in different settings,
e.g., electromagnetic emissions modeling in side-channel attacks \cite{SDattack1,SDattack2}.
In Quantum Key Distribution as well, experimental deployments still have
uncertainty limitations w.r.t. the idealized system---these limitations were especially
severe in the first deployments \cite{Brassard00}.

{\bf Jamming attacks.}
If Eve employs a jamming radio signal, she can disrupt all communication and cause a denial-of-service attack. 
However, she then reveals her presence, and homing tools can be used to geographically locate her. 
Such attacks are endemic to wireless communications, and several methods have been proposed to counteract them, 
including use of spread-spectrum, priority messages, lower duty cycle, region mapping, and mode change \cite{W1,W2}.

{\bf Artificial noise is necessary.}
When we started this work, we did not consider using artificial noise; instead,
we were planning to rely on the natural wireless channel conditions.  
We found that any node's (hence, also Eve's) erasure probability under natural network conditions can be very low (on the order of $10^{-3}$), which means that, unless we artificially increase it,
secret generation will be very slow (on the order of a few bits per second).

{\bf Moving beyond idealized modeling.}
A valid criticism of information-theoretical security is that the security proofs assume
idealized models; we also started by assuming an idealized model,
but took two further steps: (i) we attempted to artificially create network conditions that emulate the idealized model 
and (ii) we adapted our basic protocol to these network conditions.
Ours is a proof-of-concept first attempt to emulate an idealized model: we could do more careful modeling of
 the natural wireless environment and the received signals, and we could use more sophisticated antennae, as well as more
 carefully calibrated interference.  Still, we believe that our results are promising, not only in their own merit, 
but also because they indicate that it might be possible, with similar methods, to translate other
protocols that assume idealized models to practical systems.

\section{Related Work}
\label{sec:related}
$\triangleright$ {\em Information-theoretical secrecy with idealized model:}
In a seminal paper on ``wiretap'' channels, Wyner \cite{Wyner75}
pioneered the notion that one can establish information-theoretic
secrecy between Alice and Bob by utilizing the noisy broadcast nature
of wireless transmissions. However, his scheme works only with
perfect knowledge of Eve's channel and only if Eve has a
worse channel than Bob.  In a subsequent seminal work,
Maurer~\cite{Ma-IT93} showed the value of feedback from Bob to Alice,
even if Eve hears all the feedback transmissions (i.e., the feedback
channel is public).  He showed that even if the channel from Alice to
Eve is better than that to Bob, feedback allows Alice and Bob to
create a key which is information-theoretically secure from Eve.
This line of work has led to a rich set of literature on pairwise
unconditional secret-key agreement with public discussion 
(see \cite{FuzzyExt} and references therein). 
  In \cite{PDT09}, the
authors proposed increasing {\em pairwise} secrecy through friendly
jamming, but still assuming perfect knowledge of Eve's channel.
The problem of key agreement between a set of terminals with access to noisy broadcast
channel and public discussion channel (visible to the eavesdropper)
was studied in \cite{CsNa08}, where some achievable secrecy rates were
established, assuming Eve does not have access to the noisy broadcast
transmissions. This was generalized in \cite{GoAn-P2} by developing
(non-computable) outer bounds for secrecy rates.  To the best of our
knowledge, ours is the first work to consider multi-terminal secret
key agreement over erasure networks, when Eve also has access to the
noisy broadcast transmissions. Moreover, unlike the
 works in \cite{Wyner75,Ma-IT93,GoAn-P2} that
assume infinite complexity operations, our scheme is computationally
efficient.

$\triangleright$ {\em Practical protocols for pairwise secrecy:}
The fact that we establish information-theoretic secrecy for a {\em
 group} of nodes fundamentally distinguishes our work from a class of
protocols recently proposed in the literature, which aim to extract
{\em pairwise} information theoretical secrecy from physical channel
characteristics.  This class of protocols establishes secret keys
between two parties, Alice and Bob, in the presence of a passive
adversary Eve, building on different characteristics of physical
signals, such as UHF channel values \cite{unconv-key-man}, Rayleigh
fading \cite{signal-envelope}, channel impulse responses
\cite{telepathy-journal,ipsn2010}, ultra-wide band (UWB) channel properties
\cite{uwb}, and phase reciprocity \cite{phase}.  These
works leverage the reciprocity of the physical wireless channel
between Alice and Bob to establish a common secret between them;
Closer to our work are perhaps \cite{infocom2010}
and \cite{ipsn2009}, although they still generate pairwise secrecy
without interaction. The work in \cite{infocom2010} is the only
paper, as far as we know, which employs erasures; the work in
\cite{ipsn2009} has Bob deliberately introduce errors into the
transmissions of Alice to confuse Eve.  
Unlike our work, 
all these schemes 
offer modest secrecy rates (for example of the order $10$ bits/s by the
scheme in \cite{telepathy-journal}) which is orders of magnitude lower
than what our scheme can achieve. 
Moreover, they rely on the reciprocity of the channel between
Alice and Bob, 
and thus do not scale
well as the number of nodes increases.  
  Similarly, to achieve pairwise
secrecy, {a recent work \cite{EG10} proposes to artificially
  enhance the randomness of the channel fading.  
In contrast, in this work we use active interference to
  emulate noisy (erasure) channel conditions that we then utilize to
  create the appropriate channel conditions for {group secrecy}.}

$\triangleright$ {\em Computational group secrecy:}
{ Group secret key generation
with { computational
 security guarantees} has also received significant attention (see
  \cite{last,Perrig} and references therein).}

\section{Conclusions}
\label{sec:conclusion}

We have presented a protocol that enables a group of nodes, connected to the same
broadcast channel, to exchange a common secret bitstream in the presence of an adversary.
Our protocol does not use public-key (or any classic form of) cryptography and relies, instead, 
on the assumption that an eavesdropper may overhear big chunks of the communication between the other nodes,
but cannot overhear \emph{all} the bits received by any single node.
The key properties that differentiate our protocol from prior theoretical work are that
it works for an arbitrary number of nodes, has polynomial complexity,
and is implementable in simple devices without any changes to their physical or MAC layers.
On the practical side, as a proof of concept, we adapted our protocol to a small wireless testbed of $14$ m$^2$
and presented an experimental demonstration of $\numnodes = 8$ wireless nodes
generating a $38$ Kbps common secret bitstream in the presence of an adversary (located at least $1.76$ m away from any other node).
To the best of our knowledge, ours is the first experimental evidence of generating information-theoretically secret 
bits at kilobit-per-second rates.


\appendix
\subsection{Proof of Lemma~\ref{lem:secrecy}} \label{app:1}

Eve has two occasions to obtain information about the secret $\key$:
in step 1 of the initial phase, she receives $N_E$ $x$-packets;
in step 1 of the reconciliation phase, she receives the $M - L$ $z$-packets that are reliably broadcast by Alice.
According to Lemmas~\ref{lem:linear_1} and \ref{lem:linear_3}, Eve obtains no information about $\key$ in either occasion. In Lemma~\ref{lemma_conc} we prove exponential convergence to the average values.
\hfill$\blacksquare$

\begin{lemma}
\label{lem:linear_1}
Consider a set of $N$ $x$-packets, say $x_1, \ldots, x_N$, 
and assume
 Eve has a subset of size $N_E$ of the $x$-packets.
Construct $N-N_E$ $y$-packets, say $y_1,\ldots,y_M$, as
 \[ Y=A X, \]
where matrix $X$  has as rows the $N$ $x$-packets,
matrix $Y$ has as rows the $N-N_E$ $y$-packets, and 
 $A$ is the  generator
matrix of a Maximum Distance Separable (MDS) linear code 
with parameters $[N,N-N_E,N_E+1]$ ({\em e.g.,}
a Reed-Solomon code \cite{MacSloane-Coding}).
Then the $M$ $y$-packets are
information-theoretically secure from Eve, 
irrespective of which subset (of size $N_E$) of the $x$-packets  Eve has.
\end{lemma}

{\em Proof:}
$\triangleright$Let  $W$ be a matrix that has as rows the packets Eve has.
  To prove that the $y$-packets are information-theoretically secure from Eve, we must show that: $$H(Y|W) = H(Y).$$
$\triangleright$  We can write
\[\left[\begin{array}{c} Y\\ W \end{array}\right]=
\left[\begin{array}{c} A\\ A_E \end{array}\right] X \stackrel{def}{=} B X,\]
where $A_E$ is a ${N_E\times N}$ matrix of $\mathrm{rank}(A_E)=N_E$, which
specifies the $N_E$ distinct $x$-packets that are known to Eve.
 $A_E$ is {\em not known} to us, however 
we know  is that in each row of $A_E$ there is only
one $1$ and the remaining elements are zero; so all of the vectors in
the row span of $A_E$ have Hamming weight (the number of nonzero elements of a vector \cite{MacSloane-Coding})
less than or equal to $N_E$. 
On the other hand, from construction, $\mathrm{rank}(A)=N-N_E$, and each vector in the row span of $A$
has Hamming weight larger than or equal to $N_E+1$  \cite{MacSloane-Coding}; thus the row span of $A$
and $A_E$ are disjoint (except for the zero vector) and the
matrix $B$ is full-rank, i.e. $\mathrm{rank}\left( B \right)=N$.

$\triangleright$ 
If  the packets $x_i$ have length $\Lambda$, we have that:
\begin{align*}
  &H(Y|W) = H(Y,W)-H(W)= \nonumber\\
&= \mathrm{rank}\left( B \right) \Lambda - \mathrm{rank}(A_E) \Lambda = (N - N_E )\Lambda\nonumber\\
&=\mathrm{rank}(A)\Lambda=H(Y). \quad\quad \quad\quad \quad\quad\quad\quad \quad\quad\quad\quad\quad\quad  \hfill{\blacksquare}
\end{align*}

\begin{lemma}
\label{lem:linear_3}
Consider a set of $M$ $y$-packets, say $y_1, \ldots, y_M$,
and a set of $M-L$  $z$-packets, say $z_1, \ldots, z_{M-L}$, 
related as
 \[ Z=A_Z Y, \]
where matrix $Y$  has as rows the $M$ $y$-packets,
matrix $Z$ has as rows the $M-L$ $z$-packets, and 
 $A_Z$ is a known $M-L \times M$ full rank matrix. 
Assume that Eve knows all the $z$-packets. 
Using any standard basis-extension method \cite{Horn},
find an $L\times M$  matrix $A_S$, with 
$\mathrm{rank}\left( A_S \right)=L$,  such that 
$$\mathrm{rank}\left(\left[\begin{array}{c} A_S\\ A_Z \end{array}\right]\right)=M.$$
Then we can construct $L$ $s$-packets, say  $s_1,\ldots,s_L$, as 
\[\mathbf{S}=A_S Y,\]
where matrix $\mathbf{S}$  has as rows the $s$-packets,
that are information-theoretically secure from Eve.
\end{lemma} 
{\em Proof:}
To prove that the $s$-packets are information-theoretically secure from Eve, we need to show that $H(\mathbf{S}|Z) = H(\mathbf{S})$.
 Similarly to the proof of Lemma~\ref{lem:linear_1}, if  the $y$-packets have length $\Lambda$, we have that:
\begin{align*}
  &H(\mathbf{S}|Z) = H(\mathbf{S},Z)-H(Z)= \mathrm{rank}\left(\left[\begin{array}{c} A_S\\ A_Z \end{array}\right]\right)\Lambda-\nonumber\\&-  \mathrm{rank}\left( A_Z\right)\Lambda= L\Lambda= \mathrm{rank}\left( A_S\right)\Lambda=H(\mathbf{S}).\quad\quad \quad\quad\quad\hfill{\blacksquare}
\end{align*}

\begin{lemma}\label{lemma_conc}
The values of the parameters in Lemma~\ref{lem:secrecy}
converge  exponentially fast in $N$  to their expected values.
\end{lemma}
\begin{proof}
Let us consider the {\em random} variables
$M$, $M_{i}$, and $L$ defined in Section
\ref{sec:protocol}. For convenience, we will work with the normalized
random variables $\overline{M}\triangleq M/N$,
$\overline{M}_{T_i}\triangleq M_{T_i}/N$, and $\overline{L}\triangleq
L/N$.  Define the random variable $\eta_j^{(i)}$ as 
\begin{equation}
\eta_j^{(i)} = \left\{\begin{array}{ll} 1 & \text{if the $j$th $x$-packet is received}\\ & \text{by $T_i$ but not by Eve},\\ 0 & \text{otherwise}. \end{array} \right. \nonumber
\end{equation}
Then we can write $\overline{M}_{i}=\frac{1}{N}\sum_{j=1}^N \eta_j^{(i)}$ and we have 
$\mu=\mu_i\triangleq\Expc{\overline{M}_{i}}=(1-\delta)\delta_E$. As defined before,
we have also $\overline{L}=\min_i \overline{M}_{i}$.
Now the  efficiency
$E = \frac{\overline{L}}{1+\overline{M}-\overline{L}}$
calculated in Lemma~\ref{lem:efficiency} itself is a random variable,
and using the Chernoff bound we can show that it concentrates
exponentially fast in $N$ to $\frac{\delta_E
  (1-\delta)}{1+\delta_E(\delta-\Delta)}$.  It is easy to see that
$\Expc{\overline{L}}=(1-\delta)\delta_E \stackrel{def}{=}\mu$ and
$\mu_M\triangleq \Expc{\overline{M}}=(1-\delta^{m-1})\delta_E$.  Using
concentration results (Chernoff bound
\cite[Chapter~4]{MiUp-ProbComp}), we can easily show that,
$\Prob{|\overline{L}-\mu|> \epsilon\mu}\leq
\exp\left(-\Omega(n\epsilon^2 \mu N)\right)$.  Using similar argument
for $M$ we can also write $\Prob{|\overline{M}-\mu_M|>\epsilon\mu_M}
\le \exp\left(-\Omega(\epsilon^2 \mu_MN)\right)$.  Because both
$\overline{M}$ and $\overline{L}$ concentrate around their expected
values exponentially fast in $N$, so does the efficiency $E$.
\end{proof}
\subsection{Proof of Lemma \ref{lem:complexity}}

We concentrate on Alice (the other terminals perform fewer operations).
In step 3 of the initial phase, Alice creates up to $N$ $y$-packets, 
using the construction specified in lemma \ref{lem:linear_1}; this requires up to $N^3$ operations.
%
In step 1 of the reconciliation phase, she creates up to $N-L$ $z$-packets, using a network coding construction \cite{mono}; this requires up to $(N-L)N^2$ operations.
%
In step 3 of the reconciliation phase, she creates $L$ $s$-packets, using the construction
specified in Lemma~\ref{lem:linear_3}; this requires up to $L N^2$ operations. 

Thus, Alice performs at most $2N^3$ operations to create $L = \delta_E (1 - \max_i \delta_i) N$ secret packets.
For $\delta_E$ and $\delta_i$ constant, Alice performs $O(N^2)$ operations per secret packet.
\hfill $\blacksquare$

\subsection{Coded Cooperative Data Exchange}  
\label{sec:cooperative}
Assume we have $n$ nodes and a set of packets; each node has  a subset of the packets, and is interested in collecting the ones she misses. We want to achieve this using the minimum total number of transmissions from the nodes. 
We can solve this problem in {\em polynomial time}
\cite{Sprintson10,Courade10};
for completeness, we provide in the following  an Integer Linear Program (ILP)
that accepts an efficient { polynomial-time} solution  \cite{Courade10}.

We will use the following notation:
\begin{itemize}
\item $\mathcal{L}$: a subset of the nodes; there exist  $2^n$ such subsets.
\item $\mathcal{L}^c$: set of all nodes not in $\mathcal{L}$.
\item $\mathcal{P}_i^c$: set of all packets that node $i$ does not have.
\item $K_i$: total number of transmissions that source $i$ makes.
\end{itemize}
\begin{align}
& \min \;\; K_1+K_2+\ldots + K_n \nonumber\\
&\mbox{subject to}\nonumber\\
& \sum_{i\in {\mathcal{L}}} K_i \geq |\bigcap_{i\in {\mathcal{L}^c}} \mathcal{P}_{i}^c|, \quad\;\;  \forall \mathcal{L},\nonumber \\
& K_i\in \mathbf{Z}^+. \nonumber
\end{align}
\end{document}